\documentclass[12pt]{amsart}

\usepackage[latin1]{inputenc}
\usepackage{amsmath}
\usepackage{amsfonts}
\usepackage{amssymb}
\usepackage{graphics}
\usepackage{enumerate}
\usepackage{amssymb,amsmath,amsthm,amscd,latexsym,verbatim,graphicx,amsfonts}

\newtheorem{theorem}{Theorem}
\newtheorem{lemma}[theorem]{Lemma}

\newtheorem{proposition}{Proposition}

\theoremstyle{remark}

\newcommand{\bbC}{\mathbb{C}}

\newcommand{\mch}{\mathcal{H}}

\newcommand{\mcs}{\mathcal{S}}

\begin{document}
\title[ ] {Electrostatic Equilibria on the Unit Circle via Jacobi Polynomials}

\bibliographystyle{plain}

\thanks{  }


\maketitle


\centerline{Kev Johnson $\&$ Brian Simanek}


\begin{abstract}
We use classical Jacobi polynomials to identify the equilibrium configurations of charged particles confined to the unit circle.  Our main result unifies two theorems from a 1986 paper of Forrester and Rogers.
\end{abstract}


\section{Introduction}\label{intro}

The use of Jacobi polynomials to describe configurations of charged particles that are in electrostatic equilibrium goes back at least to the work of Heine and Stieltjes in the 19th century (see \cite{Heine,Stieltjes,Stieltjes1,Stieltjes2}).  Their work considered particles of identical charge confined to an interval in the real line.  The key to the calculations is to relate the condition of being a critical point of the appropriate Hamiltonian to the second order differential equation satisfied by the polynomial whose zeros mark the equilibrium points (see \cite{Szeg}).  In the case of $n$ particles confined to an interval with charged particles fixed at the endpoints, the relevant second order differential equation is precisely the ODE satisfied by the degree $n$ Jacobi polynomial $P_n^{(\alpha,\beta)}(x)$, namely
\begin{equation}\label{jacode}
(1-x^2)y''+(\beta-\alpha-(\alpha+\beta+2)x)y'+n(n+\alpha+\beta+1)y=0,
\end{equation}
where the real numbers $\alpha$ and $\beta$ are related to the magnitude of the fixed charges at the endpoints of the interval.  Many variations and generalizations of Stieltjes' work have been realized since his original papers (see for example \cite{Ismail,MMFMG}).

It was approximately 100 years before the work of Heine and Stieltjes was adapted to the setting of the unit circle by Forrester and Rogers in \cite{FR}.  In that paper, the authors studied highly symmetric configurations of charged particles that are on the unit circle and in electrostatic equilibrium, meaning the total force on each particle is normal to the circle at its location.  They described the equilibrium configurations in terms of the zeros of the appropriate Jacobi polynomials.  Our main result (Theorem \ref{better41} below) will generalize the results from that paper by allowing for a broader collection of configurations and charges.

Since we will be working with the two-dimensional electrostatic interaction, we will consider Hamiltonians of the form
\begin{equation}\label{genh}
H(\{ t_j \}_{j=1}^M)=\sum_{1\leq j<k \leq M}\sigma(e^{it_j})\sigma(e^{it_k})\log |e^{it_j}-e^{it_k}|+\sum_{b=1}^K\sum_{a=1}^M\sigma(e^{i\eta_b})\sigma(e^{it_a})\log |e^{i\eta_b}-e^{it_a}|.
\end{equation}
as in \cite{FR,Grinshpan}, where the particles at the points $\{e^{it_j}\}_{j=1}^M$ are considered ``mobile," the particles at $\{e^{i\eta_j}\}_{j=1}^K$ are considered ``fixed," and $\sigma(x)>0$ denotes the charge carried by the particle located at $x\in\bbC$.  To avoid any ambiguity that may arise from rotating the circle, we will always assume $K\geq1$.

In our main result, we will consider the configuration space that consists of all $\{e^{i\phi_j}\}_{j=1}^m$, $\{e^{i\psi_j}\}_{j=1}^m$, $\{e^{i\theta_j}\}_{j=1}^{2mn}$ such that
\begin{align*}
&0=\phi_1\\
&\phi_j<\theta_{2(j-1)n+1}<\theta_{2(j-1)n+2}<\cdots<\theta_{(2j-1)n}<\psi_j,\qquad\qquad j=1,2,\ldots,m\\
&\psi_j<\theta_{(2j-1)n+1}<\theta_{(2j-1)n+2}<\cdots<\theta_{2jn}<\phi_{j+1},\qquad\qquad\,\,\, j=1,2,\ldots,m,
\end{align*}
where $\phi_{m+1}=2\pi$.  We will denote this configuration space by $\mcs$ and note that $\mcs$ is convex.  Let us suppose that $p,q>0$ are fixed.  On $\mcs$ we consider the Hamiltonian $\tilde{H}$ given by
\begin{align*}
\tilde{H}\left(\{\phi_j\}_{j=1}^m,\{\psi_j\}_{j=1}^m,\{\theta_j\}_{j=1}^{2mn}\right)&=p\sum_{j=1}^{m}\sum_{k=1}^{2mn}\log |e^{i\phi_j}-e^{i \theta_k}|+q\sum_{j=1}^{m}\sum_{k=1}^{2mn}\log |e^{i\psi_j}-e^{i \theta_k}|\\
&+\sum_{1 \leq k<j\leq 2mn}\log |e^{i\theta_j}-e^{i\theta_k}|+p^2\sum_{1\leq k<j\leq m}\log |e^{i\phi_j}-e^{i \phi_k}|\\
&+q^2\sum_{1 \leq k<j\leq m}\log |e^{i\psi_j}-e^{i\psi_k}|+pq\sum_{1\leq k,j\leq m}\log |e^{i\phi_j}-e^{i \psi_k}|
\end{align*}
Notice that since $\phi_1=0$ always, we can think of $\tilde{H}$ as being a function of $2mn+2m-1$ real variables.  This Hamiltonian is of the form $H$ from \eqref{genh} with $K=1$ and $e^{i\eta_1}=1$, $M=2mn+2m-1$, $\sigma(e^{i\phi_j})=p$ and $\sigma(e^{i\psi_j})=q$ for all $j=1,\ldots,m$, and $\sigma(e^{i\theta_j})=1$ for all $j=1,\ldots,2mn$.  The following result is a generalization of \cite[Theorem 2.1]{FR} and \cite[Theorem 4.1]{FR}.

\begin{theorem}\label{better41}
The Hamiltonian $\tilde{H}$ attains its maximum on $\mcs$ precisely when the points $\{e^{i\phi_j}\}_{j=1}^m$ mark the $m^{th}$ roots of $1$, the points $\{e^{i\psi_j}\}_{j=1}^m$ mark the $m^{th}$ roots of $-1$, and the points $\{e^{i\theta_j}\}_{j=1}^{2mn}$ mark the zeros of the polynomial
\begin{equation}\label{polydef}
z^{mn}P_n^{(p-1/2,q-1/2)}\left(\frac{z^m}{2}+\frac{1}{2z^m}\right)
\end{equation}
\end{theorem}

An example of a configuration described in Theorem \ref{better41} can be seen in Figure 1.  The special case of Theorem \ref{better41} in which $m=1$ is precisely \cite[Theorem 2.1]{FR}.  The special case of Theorem \ref{better41} in which $p=q$ and $m$ is a power of $2$ is precisely \cite[Theorem 4.1]{FR}.  The proof of Theorem \ref{better41} will require some intermediate steps where we consider related Hamiltonians, all of the form \eqref{genh} for an appropriate choice of parameters.  More precisely, we will proceed by using the ODE \eqref{jacode} to find a critical point of the Hamiltonian $\tilde{H}$, which we will deduce is the maximizer from a uniqueness result that we prove in Section \ref{uniquesec}.  To find the critical point, we will first consider the case in which the particles with charge $p$ and $q$ are fixed (see Section \ref{fixed}) and then use a symmetry argument to handle the general case in Section \ref{mobile}.

\begin{figure}
  \centering
  \includegraphics[width=0.45\textwidth]{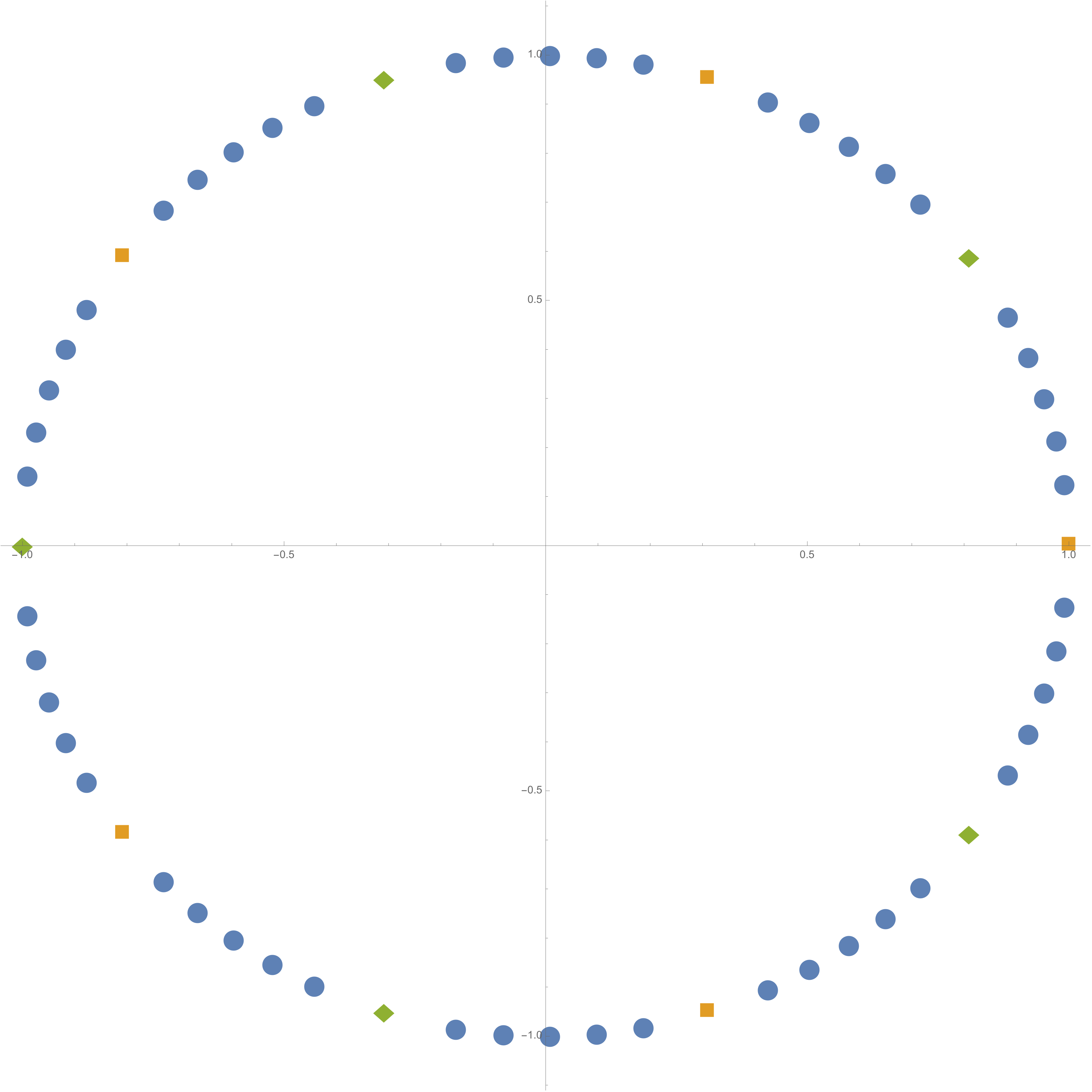}
\caption{The equilibrium configuration described by Theorem \ref{better41} when $n=5$, $m=5$, $p=2$, and $q=2.5$.  The squares mark the $5^{th}$ roots of unity, the diamonds mark the $5^{th}$ roots of $-1$, and the circles mark the roots of the polynomial \eqref{polydef}.}
\label{fig:6}       
\end{figure}

\section{Critical Points of $H$}\label{uniquesec}

Our main result of this section is a uniqueness result that applies to all Hamiltonians $H$ of the form \eqref{genh}.  We will apply it in several special cases in later sections.

\begin{theorem} \label{unique}
The Hamiltonian $H$ defined in \eqref{genh} has a unique critical point on each connected component of the domain on which $H$ is finite.
\end{theorem}

\begin{proof}
We follow the method used to prove similar statements in \cite{Ismail,MMFMG}.  Define the Hessian $\mch$ of $H$ to be
\[
\mathcal{H}_{jk}=\dfrac{\partial ^2H}{\partial t_j \partial t_k}.
\]
We will first show that $-\mathcal{H}$ is strictly positive definite, which will then imply that $H$ is strictly concave on each connected component of its domain (using the fact that every such connected component is a convex set; see \cite[Theorem 1.5]{Convex}). To this end, we calculate the partial derivatives

\begin{align*}
\dfrac{\partial H}{\partial t_k}&=\sum_{\stackrel{j=1}{j \neq k}}^M\frac{\sigma(e^{it_j})\sigma(e^{it_k})}{2}\cot \left( \dfrac{t_k-t_j}{2} \right)+\sum_{b=1}^K\frac{\sigma(e^{i\eta_b})\sigma(e^{it_k})}{2}\cot \left( \dfrac{t_k-\eta_b}{2} \right)\\
\dfrac{\partial^2 H}{\partial t_j \partial t_k}&=\dfrac{\sigma(e^{it_j})\sigma(e^{it_k})}{4}\csc ^2 \left( \dfrac{t_k-t_j}{2} \right),\qquad\qquad\qquad\qquad\qquad\qquad j\neq k\\
\dfrac{\partial^2 H}{\partial t_k^2}&=-\sum_{\stackrel{j=1}{j \neq k}}^M\dfrac{\sigma(e^{it_j})\sigma(e^{it_k})}{4}\csc ^2 \left( \dfrac{t_k-t_j}{2} \right)-\sum_{b=1}^K\frac{\sigma(e^{i\eta_b})\sigma(e^{it_k})}{4}\csc ^2 \left( \dfrac{t_k-\eta_b}{2} \right).
\end{align*}

Observe that the negative of each diagonal entry is precisely the sum of the off diagonal entries of the same row plus a positive term.  It follows that $-\mathcal{H}$ is diagonally dominant and has only positive eigenvalues and is therefore strictly positive definite.  It follows that $H$ is strictly concave on each connected component of its domain.

Notice that $H(\textbf{x})$ approaches $-\infty$ as $\textbf{x}$ approaches the boundary of a connected component of the domain of $H$.  It follows from the upper semicontinuity of $H$ that $H$ attains a maximum on every such connected component and therefore must have a critical point on every such connected component.  Uniqueness of this critical point follows from the strict concavity just proven.
\end{proof}

Recall the convention that if a particle of charge $q$ is located at a point $a\in\bbC$ and a particle of charge $p$ is located at a point $b\in\bbC$, then the force on the particle at $b$ due to the particle at $a$ is $2pq/(\bar{b}-\bar{a})$ (as in \cite{Grinshpan,Grinshpan2,ODEPOPUC}).  With this convention, we have the following lemma relating critical points of general Hamiltonians of the form \eqref{genh} to the condition of electrostatic equilibrium (see also \cite{Grinshpan}).

\begin{lemma}\label{equiv}
For the Hamiltonian $H$ from \eqref{genh}, it holds that
\[
\frac{\partial H}{\partial t_k}(\{t_j^*\}_{j=1}^{M})=0
\]
if and only if
\begin{equation}\label{lemeq}
\sum_{{j=1}\atop{j\neq k}}^{M}\frac{\sigma(e^{it_j^*})\sigma(e^{it_k^*})}{e^{it_k^*}-e^{it_j^*}}+\sum_{j=1}^{K}\frac{\sigma(e^{i\eta_j})\sigma(e^{it_k^*})}{e^{it_k^*}-e^{i\eta_j}}=e^{-it_k^*}\sigma(e^{it_k^*})\left(\sum_{\stackrel{j=1}{j \neq k}}^M\frac{\sigma(e^{it_j^*})}{2}+\sum_{b=1}^K\frac{\sigma(e^{i\eta_b})}{2}\right).
\end{equation}
\end{lemma}

\begin{proof}
We have already seen that
\[
\dfrac{\partial H}{\partial t_k}=\sum_{\stackrel{j=1}{j \neq k}}^M\frac{\sigma(e^{it_j})\sigma(e^{it_k})}{2}\cot \left( \dfrac{t_k-t_j}{2} \right)+\sum_{b=1}^K\frac{\sigma(e^{i\eta_b})\sigma(e^{it_k})}{2}\cot \left( \dfrac{t_k-\eta_b}{2} \right)
\]
We can rewrite this as
\[
\dfrac{\partial H}{\partial t_k}=\sum_{\stackrel{j=1}{j \neq k}}^M\frac{\sigma(e^{it_j})\sigma(e^{it_k})e^{it_k}}{e^{it_k}-e^{it_j}}-\sum_{\stackrel{j=1}{j \neq k}}^M\frac{\sigma(e^{it_j})\sigma(e^{it_k})}{2}+\sum_{b=1}^K\frac{\sigma(e^{i\eta_b})\sigma(e^{it_k})e^{it_k}}{e^{it_k}-e^{i\eta_b}}-\sum_{b=1}^K\frac{\sigma(e^{i\eta_b})\sigma(e^{it_k})}{2}
\]
and the desired result follows.
\end{proof}

It follows from Lemma \ref{equiv} that $\{t_j^*\}_{j=1}^M$ is a critical point of $H$ if and only if we have equality in \eqref{lemeq} for all $k=1,2,\ldots,M$.  For future reference, notice that the expression
\begin{equation}\label{csum}
\sum_{\stackrel{j=1}{j \neq k}}^M\frac{\sigma(e^{it_j^*})}{2}+\sum_{b=1}^K\frac{\sigma(e^{i\eta_b})}{2}
\end{equation}
on the right-hand side of \eqref{lemeq} is one half of the sum of the charges on all of the particles in the system except the one at $e^{it_k^*}$.

\section{$p$ and $q$ Charges Fixed}\label{fixed}

In this section, we will take a preliminary step towards the proof of Theorem \ref{better41} and consider the Hamiltonian $\hat{H}$ given by
\begin{align*}
&\hat{H}\left(\{\theta_j\}_{j=1}^{2mn}\right)\\
&\qquad=p\sum_{j=1}^{m}\sum_{k=1}^{2mn}\log |e^{2\pi ij/m}-e^{i \theta_k}|+q\sum_{j=1}^{m}\sum_{k=1}^{2mn}\log |e^{(2j+1)i\pi/m}-e^{i \theta_k}|+\sum_{1 \leq k<j\leq 2mn}\log |e^{i\theta_j}-e^{i\theta_k}|
\end{align*}
The Hamiltonian $\hat{H}$ isolates the $\theta$-dependence of the Hamiltonian $\tilde{H}$ by fixing the locations of the points $\{e^{i\phi_j}\}_{j=1}^m$ and $\{e^{i\psi_j}\}_{j=1}^m$ at the $m^{th}$ roots of $1$ and $-1$ respectively.  Let us also define the configuration space $\hat{\mcs}$ to be the set of all $\{\theta_j\}_{j=1}^{2mn}$ satisfying
\begin{align*}
\theta_j&<\theta_{j+1}\qquad\qquad\qquad j=1,2,\ldots,2mn-1,\\
\frac{(k-1)\pi}{m}&<\theta_j<\frac{k\pi}{m},\qquad\qquad\mbox{if}\qquad (k-1)n<j\leq kn.
\end{align*}
Observe that $\hat{\mcs}$ is convex.  In this context, we have the following result.

\begin{proposition}\label{better41a}
The unique configuration that maximizes $\hat{H}$ on $\hat{\mcs}$ occurs when the points $\{e^{i\theta_j}\}_{j=1}^{2mn}$ mark the zeros of the polynomial in \eqref{polydef}.
\end{proposition}

\begin{proof}
Notice that the Hamiltonian $\hat{H}$ is of the form $H$ from \eqref{genh} with $M=2mn$; $K=2m$; $\{e^{i\eta_j}\}_{j=1}^{2m}$ equal to the roots of $z^{2m}-1$; $\sigma(e^{i\theta_k})=1$ for all $k=1,\ldots,2mn$; $\sigma(e^{2\pi ij/m})=p$; and $\sigma(e^{(2j+1)i\pi/m})=q$ for all $j=1,\ldots,m$.  By Theorem \ref{unique}, it will suffice to show that the zeros of the polynomial in \eqref{polydef} form a critical point of $\hat{H}$ on $\hat{\mcs}$ because the maximum must occur at a critical point.

Let us define the polynomial $Q_{nm}(z)$ to be the polynomial given in \eqref{polydef}.
It follows that (where we abbreviate $P_n^{(\alpha,\beta)}$ by $P_n$)
\begin{align*}
Q_{nm}'(z)&=nmz^{nm-1}P_n\left(\frac{z^m}{2}+\frac{1}{2z^m}\right)+mz^{nm-1}\left(\frac{z^m}{2}-\frac{1}{2z^{m}}\right)P_n'\left(\frac{z^m}{2}+\frac{1}{2z^m}\right)\\
Q_{nm}''(z)&=nm(nm-1)z^{nm-2}P_n\left(\frac{z^m}{2}+\frac{1}{2z^m}\right)+(2nm^2-m)z^{nm-2}\left(\frac{z^m}{2}-\frac{1}{2z^{m}}\right)P_n'\left(\frac{z^m}{2}+\frac{1}{2z^m}\right)\\
&\qquad +m^2z^{mn-2}\left(\frac{z^m}{2}+\frac{1}{2z^m}\right)P_n'\left(\frac{z^m}{2}+\frac{1}{2z^m}\right)+m^2z^{mn-2}\left(\frac{z^m}{2}-\frac{1}{2z^{m}}\right)^2P_n''\left(\frac{z^m}{2}+\frac{1}{2z^m}\right)
\end{align*}
Using these calculations and the differential equation \eqref{jacode}, one can verify that $Q_{nm}$ satisfies the ODE
\[
y''+\left[mz^{m-1}\left(\frac{2\alpha+1}{z^m-1}-\frac{\alpha+\beta+2n}{z^m}+\frac{2\beta+1}{z^m+1}\right)-\frac{m-1}{z}\right]y'(z)+(mz^{m-1})^2T(z^m)y(z)=0,
\]
where
\[
T(z)=-n\left[\frac{\alpha+\beta+1}{z^2}-\frac{2(\beta-\alpha)z-(\alpha+\beta+1)(z^2+1)}{z^2(z^2-1)}\right]
\]
We see that the only poles of $T$ are at $0$, $1$ and $-1$ and hence the zeros of $Q_{nm}$ and the poles of $T(z^m)$ are disjoint sets.  We also notice that
\[
m(2\alpha+1)\frac{z^{m-1}}{z^m-1}=\sum_{j=1}^m\frac{2\alpha+1}{z-e^{2\pi ij/m}},\qquad\qquad m(2\beta+1)\frac{z^{m-1}}{z^m+1}=\sum_{j=1}^m\frac{2\beta+1}{z-e^{\pi i(2j+1)/m}}
\]
Thus, if $Q_{nm}(e^{i\theta_j^*})=0$ for $j=1,2,\ldots,2mn$, then it holds that
\begin{equation}\label{equ1}
\sum_{{k=1}\atop{k\neq j}}^{2mn}\frac{2}{e^{i\theta_j^*}-e^{i\theta_k^*}}+\sum_{j=1}^m\frac{2\alpha+1}{e^{i\theta_j^*}-e^{i\phi_j}}-\frac{m(\alpha+\beta+2n)+m-1}{e^{i\theta_j^*}}+\sum_{j=1}^m\frac{2\beta+1}{e^{i\theta_j^*}-e^{i\psi_j}}=0,
\end{equation}
where we used the identity
\[
\frac{Q''(e^{i\theta_j^*})}{Q'(e^{i\theta_j^*})}=\sum_{{k=1}\atop{k\neq j}}^{2mn}\frac{2}{e^{i\theta_j^*}-e^{i\theta_k^*}}.
\]
Now set $p=\alpha+1/2$ and $q=\beta+1/2$ and calculate the expression in \eqref{csum}.  We find
\[
\sum_{{k=1}\atop{k\neq j}}^{2mn}\frac{\sigma(e^{i\theta_k^*})}{2}+\sum_{k=1}^{m}\frac{\sigma(e^{2ki\pi/m})}{2}+\sum_{k=1}^{m}\frac{\sigma(e^{(2k+1)i\pi/m})}{2}=\frac{2mn-1+m(\alpha+\beta+1)}{2}
\]
Since this is true for every $j=1,\ldots,2mn$, Lemma \ref{equiv} and \eqref{equ1} show that the zeros of $Q_{nm}$ form a critical point of $\hat{H}$.  By Theorem \ref{unique}, this is the only critical point on $\hat{\mcs}$ and hence is the maximizing configuration.
\end{proof}

\section{$p$ and $q$ Charges Mobile}\label{mobile}

Now we will consider the full Hamiltonian $\tilde{H}$.  This Hamiltonian is of the form $H$ with $K=1$, $\eta_1=0$ and $M=2mn+2m-1$, and with the $t_j's$ denoting the arguments of all of the particles in the system other than the particle at $1$.

\begin{proof}[Proof of Theorem \ref{better41}]
By Theorem \ref{unique}, it suffices to show that the suggested configuration is a critical point of $\tilde{H}$.  Let $\{e^{i\theta_j^*}\}_{j=0}^{2mn}$ denote the zeros of $Q_{nm}$ from the previous section.  We know from Proposition \ref{better41a} that
\[
\frac{\partial}{\partial\theta_k}\tilde{H}\left(\left\{\frac{2j\pi}{m}\right\}_{j=0}^{m-1},\left\{\frac{(2j+1)\pi}{m}\right\}_{j=0}^{m-1},\{\theta_j^*\}_{j=1}^{2mn}\right)=0
\]
for all $k=1,2,\ldots,2mn$.  It remains to check that the partial derivatives with respect to each $\phi_k$ ($k=2,\ldots,m$) and $\psi_k$ ($k=1,\ldots,m$) vanish at this configuration and for this we will use Lemma \ref{equiv}.

From symmetry, we know that in this configuration, the sum of the forces on each particle of charge $p$ or $q$ is radial at that point.  Also, the magnitude of the force on all particles of charge $p$ is the same and the magnitude of the force on all particles of charge $q$ is the same.  This means that if $e^{i\phi_j}=e^{2i\pi(j-1)/m}$ and $e^{i\psi_j}=e^{i\pi(2j-1)/m}$, then for each $k=1,2,\ldots,m$ there are real constants $C$ and $C'$ so that
\begin{align*}
\sum_{j=1}^{2mn}\frac{2p}{e^{i\phi_k}-e^{i\theta_j^*}}+\sum_{j=1}^{m}\frac{2pq}{e^{i\phi_k}-e^{i\psi_j}}+\sum_{{j=1}\atop{j\neq k}}^m\frac{2p^2}{e^{i\phi_k}-e^{i\phi_j}}&=\frac{C}{e^{i\phi_k}}\\
\sum_{j=1}^{2mn}\frac{2q}{e^{i\psi_k}-e^{i\theta_j^*}}+\sum_{j=1}^m\frac{2pq}{e^{i\psi_k}-e^{i\phi_j}}+\sum_{{j=1}\atop{j\neq k}}^m\frac{2q^2}{e^{i\psi_k}-e^{i\psi_j}}&=\frac{C'}{e^{i\psi_k}}
\end{align*}
We can rewrite these expressions as
\begin{align*}
2p\frac{Q_{nm}'(e^{i\phi_k})}{Q_{nm}(e^{i\phi_k})}+2pq\frac{D_{m}'(e^{i\phi_k})}{D_{m}(e^{i\phi_k})}+p^2\frac{B_{m}''(e^{i\phi_k})}{B_{m}'(e^{i\phi_k})}&=\frac{C}{e^{i\phi_k}}\\
2q\frac{Q_{nm}'(e^{i\psi_k})}{Q_{nm}(e^{i\psi_k})}+2pq\frac{B_{m}'(e^{i\psi_k})}{B_{m}(e^{i\psi_k})}+q^2\frac{D_{m}''(e^{i\psi_k})}{D_{m}'(e^{i\psi_k})}&=\frac{C'}{e^{i\psi_k}},
\end{align*}
where $B_m(z)=z^m-1$ and $D_m(z)=z^m+1$.  The above expressions simplify to
\begin{align*}
2pnm+pqm+p^2(m-1)&=C\\
2qnm+pqm+q^2(m-1)&=C'
\end{align*}
This shows that we can write
\begin{align}
\label{phieq}\sum_{j=1}^{2mn}\frac{p}{e^{i\phi_k}-e^{i\theta_j^*}}+\sum_{j=1}^m\frac{pq}{e^{i\phi_k}-e^{i\psi_j}}+\sum_{{j=1}\atop{j\neq k}}^m\frac{p^2}{e^{i\phi_k}-e^{i\phi_j}}&=e^{-i\phi_k}\frac{2pnm+pqm+p^2(m-1)}{2}\\
\label{psieq}\sum_{j=1}^{2mn}\frac{q}{e^{i\psi_k}-e^{i\theta_j^*}}+\sum_{j=1}^m\frac{pq}{e^{i\psi_k}-e^{i\phi_j}}+\sum_{{j=1}\atop{j\neq k}}^m\frac{q^2}{e^{i\psi_k}-e^{i\psi_j}}&=e^{-i\psi_k}\frac{2qnm+pqm+q^2(m-1)}{2}
\end{align}
for all $k=1,\ldots,m$.  At $e^{i\phi_k}$, the sum of the charges on all of the other particles is $2mn+mq+(m-1)p$.  At $e^{i\psi_k}$, the sum of the charges on all of the other particles is $2mn+mp+(m-1)q$.

We can now apply Lemma \ref{equiv} to conclude that the suggested configuration is a critical point of $\tilde{H}$ on $\mcs$ (note that to reach this conclusion, we do not need to apply (\ref{phieq}) when $k=1$ because we assume $\phi_1=0$ always).  By Theorem \ref{unique}, this is the only critical point in $\mcs$ and hence must be the maximizing configuration.
\end{proof}



\begin{thebibliography}{99}














\bibitem{FR} P. Forrester and J. B. Rogers, {\em Electrostatics and the zeros of the classical polynomials}, SIAM J. Math. Anal. 17 (1986), 461--468.

\bibitem{Grinshpan} A. Grinshpan, {\em A minimum energy problem and Dirichlet spaces}, Proc. Amer. Math. Soc. 130 (2002), no. 2, 453--460.

\bibitem{Grinshpan2} A. Grinshpan, {\em Electrostatics, hyperbolic geometry and wandering vectors},  J. London Math. Soc. (2) 69 (2004), no. 1, 169--182.

\bibitem{Heine} E. Heine, {\em Handbruch der Kugelfunktionen}, vol. II, second ed. G. Reimer, Berlin, 1878.

\bibitem{Ismail} M. E. H. Ismail, {\em An electrostatics model for zeros of general orthogonal polynomials}, Pacific J. Math. 193 (2000), no. 2, 355--369.


\bibitem{MMFMG} F. Marcell\'{a}n, A. Mart\'{i}nez-Finkelshtein, and P. Mart\'{i}nez-Gonz\'{a}lez, {\em Electrostatic models for zeros of polynomials: old, new, and some open problems}, J. Comput. Appl. Math. 207 (2007), no. 2, 258--272.


\bibitem{ODEPOPUC} B. Simanek, {\em An electrostatic interpretation of the zeros of paraorthogonal polynomials on the unit circle}, SIAM J. Math. Anal. 48 (2016), no. 3, 2250--2268.

\bibitem{Convex} B. Simon, {\em Convexity: An analytic viewpoint}, Cambridge Tracts in Mathematics 187, Cambridge University Press, Cambridge, 2011.

\bibitem{Stieltjes} T. Stieltjes, {\em Sur certains polyn\^{o}mes qui v\'{e}rifient une \'{e}quation diff\'{e}rentielle lin\'{e}aire du second ordre et sur la theorie des fonctions de lam\'{e}}, Acta Math. 6 (1885), no. 1, 321--326.

\bibitem{Stieltjes1} T. Stieltjes, {\em Sur quelques th\'{e}or\`{e}ms d'alg\`{e}bre}, Comptes Rendus de l'Academie des Sciences, Paris 100 (1885), 439--440; Oeuvres Compl\`{e}tes, Vol. 1, 440--441.

\bibitem{Stieltjes2}  T. Stieltjes, {\em Sur les polyn\^{o}mes de Jacobi}, Comptes Rendus de l'Academie des Sciences, Paris, 100 (1885), 620--622; Oevres Compl\`{e}tes, Vol. 1, 442--444.

\bibitem{Szeg} G. Szeg\H{o}, {\em Orthogonal Polynomials}, Fourth Edition, American Mathematical Society, Colloquium Publications, Vol. XXIII, Providence, RI, 1975.











\end{thebibliography}
\end{document}